\documentclass[a4paper]{article}
\usepackage[margin=1in]{geometry}
\usepackage[utf8]{inputenc}
\usepackage{amsmath}
\usepackage{amsfonts}
\usepackage{natbib}
\usepackage{graphicx}
\usepackage{float}
\graphicspath{{./figures/}}
\usepackage{appendix}
\usepackage{amsthm}
\usepackage{hyperref}

\title{Identification and Inference with Min-over-max Estimators for the Measurement of Labor Market Fairness}
\author{Karthik Rajkumar\footnote{Email \href{mailto:krajkumar@linkedin.com}{krajkumar@linkedin.com}. I thank Guillaume Saint-Jacques, Kenneth Tay and YinYin Yu for feedback on earlier drafts, and Parvez Ahammad for his support of this work.}\\LinkedIn Corporation}

\newtheorem{prop}{Proposition}

\begin{document}

\maketitle

\begin{abstract}
    These notes shows how to do inference on the Demographic Parity (DP) metric. Although the metric is a complex statistic involving min and max computations, we propose a smooth approximation of those functions and derive its asymptotic distribution. The limit of these approximations and their gradients converge to those of the true max and min functions, wherever they exist. More importantly, when the true max and min functions are not differentiable, the approximations still are, and they provide valid asymptotic inference everywhere in the domain. We conclude with some directions on how to compute confidence intervals for DP, how to test if it is under 0.8 (the U.S. Equal Employment Opportunity Commission fairness threshold), and how to do inference in an A/B test.  
\end{abstract}

\section{Introduction}
Min-over-max style estimators arise naturally in problems studying disparities across groups. Consider the estimator that we call Demographic Parity (DP), which is the ratio of a certain outcome in the lowing-performing group to the same in the highest-performing group. One might compute such an estimator on, say, application response rates across groups where equity is desired. It is easy to compute and interpret: DP is less than or equal to 1 and the closer it is to unity, the closer we are to equity between groups. The DP metric makes equity audits simple: a value below 0.8 is indicative of inequity and calls for further investigation. 

Because the min and max functions are not differentiable everywhere, min-over-max estimators run into certain issues with statistical inference and asymptotic normality is not available. In these notes, we take an approximation approach to min-over-max estimators. Using smooth and differentiable approximations to the min and max functions, where the level of approximation is decided by a tuning parameter, we control the skewness of the asymptotic distributions of the ratio estimators and recover normality. Since larger values of the tuning parameter mean better approximations of our constituent functions, but also more asymptotic skewness, we provide upper bounds on it as a function of the sample size of the dataset. In this way, as datasets get larger and larger, the approximation matters less and we can get closer to the true ratio estimator, while preserving statistical normality.

The rest of the notes is organized as follows. Section \ref{sec:problem} motivates the problem with the Demographic Parity estimator and formalizes the mathematical notation. Section \ref{sec:asymptotics-sample-means} provides the basic asymptotic result we build on. Section \ref{sec:smooth-max} introduces the approximation trick that is core to this work and Section \ref{sec:asymptotic-distr-min-max} quantifies the approximation error in their asymptotic distribution. Section \ref{sec:inference-on-DP} presents the central asymptotic result with the approximation estimators, along with how to do inference in practical settings including how to tune the approximation parameter. Sections \ref{sec:other-considerations} and \ref{sec:small-sample-critique} go over extensions and limitations of our approximation strategy, and Section \ref{sec:conclusion} concludes. 

\section{Problem setup}
\label{sec:problem}
\begin{itemize}
    \item We have multiple groups, $g = 1, 2, ..., G$.
    \item Each group is associated with a success probability (such as the probability of landing a job given one application). We call this $s_g \in [0, 1]$.
    \item Each unit from group $g$ is sampled with a probability $p_g$. That is, sampling of units from groups happens with multinomial probabilities $\left\{p_g \right\}$, where $\sum_g p_g = 1$. 
    \item A unit here refers to one application. Implicitly, we model the number of applications by each \textit{individual} as random.
    \item Each unit realizes a binary outcome, $Y_i$, which depends on the success probability of its respective group. 
    \item The total number of units in the sample is $N$. 
    \item We obtain maximum-likelihood estimates (sample means) of success probabilities for each group, $\hat{s}_g = \frac{\sum_{i \in g} Y_i}{\sum_{i \in g} 1}$. 
\end{itemize}

Given this setup, our metric of interest, Demographic Parity (DP for short), is defined as 
\begin{equation*}
    \hat{\text{DP}} = \frac{\min_g \hat{s}_g}{\max_g \hat{s}_g},
\end{equation*}

and the theoretical estimand is 
\begin{equation}
    \text{DP}_0 = \frac{\min_g s_g}{\max_g s_g}.
\end{equation}

Note that, by definition, both $\text{DP}_0$ and $\hat{\text{DP}}$ are limited to the interval $[0, 1]$ since we have $0 \leq \min_g \left\{ a_g \right\} \leq \max_g \left\{ a_g \right\} \leq 1$ for any set of nonnegatives $\left\{ a_g \right\}_{g=1}^G$. 

\section{Asymptotics of sample means}
\label{sec:asymptotics-sample-means}

  The central limit theorem gives us 
  
  \begin{equation} \label{clt_means}
      \sqrt{N} \left( 
      \begin{bmatrix} \hat{s}_1 \\ \vdots \\ \hat{s}_G \end{bmatrix} - 
      \begin{bmatrix} s_1 \\ \vdots \\ s_G \end{bmatrix} \right) \stackrel{d}{\longrightarrow}
      \mathcal{N} \left(0, 
      \begin{bmatrix} \sigma_1^2 & & 0 \\ & \ddots & \\ 0 & & \sigma_G^2 \end{bmatrix} \right),
  \end{equation}
where $\sigma_g^2 = \frac{s_g (1 - s_g)}{p_g}$.

This variance-covariance matrix can be easily estimated with regression of binary outcomes on group fixed effects and heteroskedastic standard errors. The command to use for those in R is:
\begin{verbatim}
    vcovHC(type="HC0")
\end{verbatim}

\section{The smooth maximum function}
\label{sec:smooth-max}
To get the asymptotic distribution of the max and min of these sample means, we would typically resort to a delta method. However, the delta method requires a continuously differentiable function, which the max and min are not. For this reason, we use a ``smooth'' version of the max function, which is 
\begin{equation}
    \text{rsmax}_{\alpha}(a_1, ..., a_G) = \frac{1}{\alpha} \log \left( \sum_{j = 1}^G e^{\alpha a_j} \right),
\end{equation}

for a parameter $\alpha > 0$. Here $\text{rsmax}$ stands for ``real softmax'' \citep{zhang_lipton_li_smola}. This approximation function is also known as the LogSumExp function. The function has the nice property that
\begin{equation*}
    \lim_{\alpha \rightarrow \infty} \text{rsmax}_\alpha (\left\{ a_j \right\}) = \max (\left\{ a_j \right\}),
\end{equation*}
while maintaining differentiability. Figure \ref{fig:rsmax-approx} provides a visualization of the approximation. 

\begin{figure}[h]
    \centering
    \includegraphics[scale=0.4]{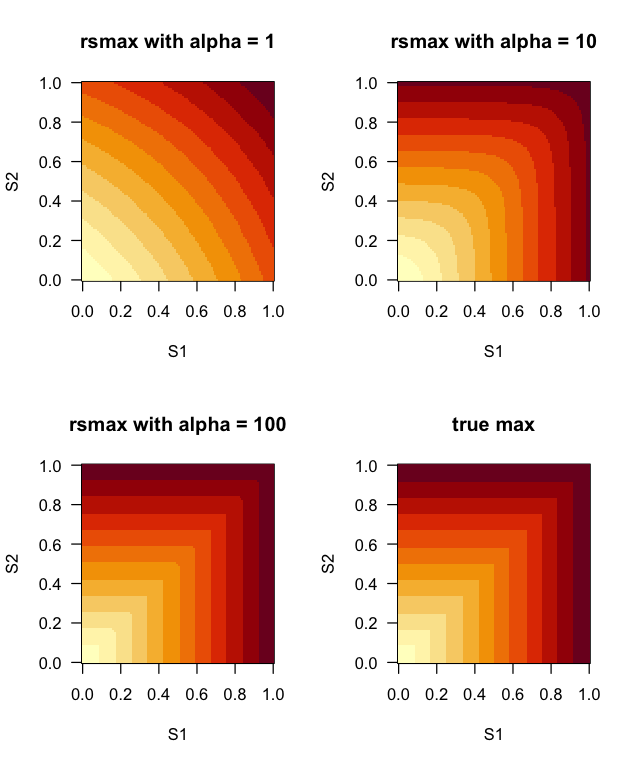}
    \caption{rsmax approximates the true max function as $\alpha \rightarrow \infty$. The lines represent points whose function values are equal.}
    \label{fig:rsmax-approx}
\end{figure}

Its gradient is 

\begin{equation}
    \text{softmax}_\alpha (\left\{ a_j \right\}) = \left\{ \frac{e^{\alpha a_j}}{\sum_i e^{\alpha a_i}} \right\}.
\end{equation}

The interpretation of the gradient is very simple when we look at the limiting case of $\alpha \rightarrow \infty$:
\begin{itemize}
    \item If there is a unique maximum among $\left\{ a_j \right\}$, the gradient vector is zero everywhere except for the element associated with the maximum, where it is 1. 
    \item If there are multiple equal maxima, the gradient vector is zero everywhere except for those elements associated with the maximum, where $\text{softmax}$ places equal weight on each maximum, i.e. $1/\text{\# maxima}$.
\end{itemize}

\subsection{The smooth minimum function}
The smooth minimum is the min version of the smooth maximum, and is defined as 
\begin{equation*}
    \text{rsmin}_\alpha (\left\{ a_j \right\} ) := -\text{rsmax}_{\alpha} (\left\{ -a_j \right\} ),
\end{equation*}
for $\alpha > 0$. As $\alpha \rightarrow \infty$, it converges to the true $\min$, while maintaining differentiability throughout. 

Its gradient is 
\begin{equation*}
    \text{softmin}_\alpha (\left\{ a_j \right\}) = \text{softmax}_{\alpha} (\left\{ -a_j \right\}),
\end{equation*}
and the interpretation of the gradient is that it is 1 at the unique minimum and 0 everywhere else. As usual, in the event of multiple equal minima, it places equal weight on each of them (still with 0 elsewhere). 

\subsection{Softmax computational stability}

Computing the functions rsmax and softmax has problems with floating point overflow, even though each $\hat{s}_j \in [0, 1]$ because we scale them by $\alpha$, which is a large number. Since we perform an expoentiating operation, even for moderate values of $\alpha$, the term $e^{\alpha s_j}$ could result in overflow. \\

To deal with this problem, we shift all inputs to the functions by their largest value \citep{blanchard2019accurate}. That is, we compute them as 
\begin{align*}
    \text{rsmax}_\alpha \left\{ a_j \right\} & = a_{(G)} + \frac{1}{\alpha} \log \left(\sum_g e^{\alpha(a_g - a_{(G)})} \right) \\
    \text{softmax}_\alpha \left\{ a_j \right\} & = \left\{ \frac{e^{\alpha (a_i - a_{(G)})}}{\sum_g e^{\alpha (a_g - a_{(G)})}} \right\},
\end{align*}
for $a_{(G)} = \max_g a_g$. This shift ensures every term $\alpha(a_j - a_{(G)}) \leq 0$. Thus, the exponential terms are always bounded by 1, and therefore the computation of the functions is numerically stable. \\

Given, these numerically stable rsmax and softmax functions, we can compute the corresponding rsmin and softmin functions as usual as \begin{align*}
    \text{rsmin}_\alpha \left\{ a_j \right\}  & = -\text{rsmax}_{\alpha} (\left\{ -a_j \right\} ) \\
    \text{softmin}_\alpha \left\{ a_j \right\} & = \text{softmax}_{\alpha} (\left\{ -a_j \right\}).
\end{align*}

\section{Asymptotic distributions of $\min$ and $\max$ of sample means}
\label{sec:asymptotic-distr-min-max}

Applying the delta method to equation \ref{clt_means} with function 
\begin{equation*}
    h(\left\{ a_j \right\}) = \begin{bmatrix} \text{rsmin}_\alpha \left\{ a_j \right\} \\
    \text{rsmax}_\alpha \left\{ a_j \right\} \end{bmatrix},
\end{equation*}

we get 
\begin{equation}
    \sqrt{N} \left( 
    \begin{bmatrix} \text{rsmin}_\alpha \left\{ \hat{s}_j \right\} \\
    \text{rsmax}_\alpha \left\{ \hat{s}_j \right\} \end{bmatrix} - 
    \begin{bmatrix} \text{rsmin}_\alpha \left\{ s_j \right\} \\
    \text{rsmax}_\alpha \left\{ s_j \right\} \end{bmatrix}\right) \stackrel{d}{\longrightarrow} 
    \mathcal{N} \left(0, \nabla ^\top h \cdot \Sigma \cdot \nabla h \right),
    \label{clt-minmax}
\end{equation}

where $\Sigma$ is the variance-covariance matrix from equation \ref{clt_means}, and $\nabla h$ is the Jacobian of function $h$ evaluated at the true means, $\left\{ a_j \right\}$. \\

The interpretation of this new variance-covariance matrix is straightforward. Consider the limiting case $\alpha \rightarrow \infty$, for the sake of exposition:
\begin{itemize}
    \item Suppose there is one unique max and one unique min among $\left\{ a_j \right\}$. Then the first column of $\nabla h$ is one at the min and the second column is one at the max, with the rest of the elements perfectly 0. 
    \item Thus, the variance-covariance matrix picks the true min sample mean and the true max sample mean as:
    \[
        \begin{bmatrix}
        \sigma^2_{\min} & 0 \\
        0 & \sigma^2_{\max} 
        \end{bmatrix},
    \]
    where $\sigma^2_{\min} = \sigma^2_i$ for $i$ the true min (likewise for the max).
    \item In the case of multiple max or mins (where not all elements are identical, that is max and min do not overlap), $\sigma^2_{\min}$ is an equally-weighted average of all the variables that are min. (Ditto for min.) There is still no covariance term because the max and min variables are different.
    \item In the final case where all values are identical (to be clear we mean values of $\left\{ s_j \right\}$), we have the above weighted variances, along with a covariance term, as follows:
    
    \begin{equation*}
    \frac{\sum_j \sigma_j^2}{G^2}
        \begin{bmatrix}
        1 & 1 \\
        1 & 1 \\
        \end{bmatrix}.
    \end{equation*}
    This follows from the fact that the min and the max are exactly identical now, and their estimator is an equally weighted average of each sample mean like so:  $\frac{1}{G} \sum_j \hat{s}_j$. 
\end{itemize}

Remember that the gradient of the true max function does not exist when two (or more) elements of the input vector are equal. However, the limit of the gradient of the real softmax function, i.e. the limit of the softmax function as $\alpha \rightarrow \infty$ \textit{does} exist. 

\subsection{Swapping the order of the limits $N \rightarrow \infty$ and $\alpha \rightarrow \infty$}
The true max and min functions are differentiable except at inputs where there are multiple argmaxes and argmins respectively. In that sense, we could directly apply the delta method on the true max and min functions and obtain the same asymptotic distribution as the ones derived above, because the limit of the softmax and softmin functions is identical to the gradients of the true max and min functions. The caveat is that this equivalence only works at inputs where there aren't multiple argmaxes (or argmins) in $\left\{ s_j \right\}$. \\

What about in the cases when we do have multiple argmaxes (or argmins) in $\left\{ s_j \right\}$? The \textit{practical} answer is that if we used a small enough $\alpha$, we have a smooth enough function and have perfect inference even in these edge cases. However, the lower the $\alpha$, the worse is the approximation of the true max function. Thus there is a tradeoff between having asymptotical normality (at low $\alpha$) and inferring the right estimand at the cost of asymptotic bias and non-normality (at high $\alpha$). See Appendix \ref{appendix-DP-nondifferentiable} for more on non-differentiability. \\

We now present a formal analysis of the approximation error when using Taylor expansions in the delta method. 

\subsection{Approximation error in the delta method}
We now quantify the error in our asymptotic distribution from using our approximation functions and provide guidance on how to choose the approximation parameter $\alpha$ that trades off the need for a better approximation with the stability properties that come from a well-behaved estimator.

The CLT of Bernoulli sample means gives us 
\begin{equation*}
    \sqrt{N} (\hat{s} - s) \stackrel{d}{\longrightarrow} \mathcal{N} (0, \Sigma),
\end{equation*}
for vector valued $\hat{s}$ and $s$. Consider a function $g$. In our context, $g$ could be the rsmax function or the rsmin. A first-order Taylor expansion of the estimate around the true mean gives us 
\begin{equation*}
    g(\hat{s}) = g(s) + \frac{\nabla g(s)}{1!} (\hat{s} - s) + \underbrace{R_1(\hat{s})}_{\text{Remainder}},
\end{equation*}
where $R_1$ is the remainder from the first-order approximation and is equal to $\frac{\nabla^2 g(s^\prime)}{2!} (\hat{s} - s)^2$ for some $s^\prime$ between $s$ and $\hat{s}$ (a convex combination of the two). Thus, our asymptotic distribution has two terms:
\begin{equation*}
    \sqrt{N} \left(g(\hat{s}) - g(s) \right) = \underbrace{\nabla g(s) \cdot \sqrt{N} (\hat{s} - s)}_{O_p(1) \text{ from CLT above}} + \sqrt{N} R_1(\hat{s}).
\end{equation*}

The remainder (approximation error) term, $\sqrt{N}R_1$, is 
\begin{equation*}
\begin{split}
    & \sqrt{N} \cdot \underbrace{\frac{\nabla^2 g(s^\prime)}{2}}_{\text{Assume $O_p(N^q)$}} \cdot \underbrace{(\hat{s} - s)^2}_{O_p(\frac{1}{N})} \\
    & = O_p(N^{1/2 + q -1}) \\
    & = O_p(N^{q - 1/2}).
\end{split}
\end{equation*}
Thus, if we want the approximation error to die down with $N$, we require $q < 1/2$. But what is $q$? Remember that it is the probabilistic order of the Hessian of our function $g$. Consider the case of the rsmax function (similar arguments apply to the rsmin). Its Hessian is the derivative of the softmax function. Denote the softmax function evaluated on the $i$th element as $\mathcal{S}_i$. Then its derivative is given by 
\begin{equation*}
    \frac{\partial}{\partial a_j} \mathcal{S}_i = \alpha \underbrace{\mathcal{S}_i (\delta_{ij} - \mathcal{S}_j)}_{\text{Bounded in [-1, 1]}}.
\end{equation*}
Here $\delta_{ij}$ is the Kronecker delta. Thus $\alpha$ determines the error rate of Hessian and from the above analysis, we determine that 
\begin{equation}
    \alpha = O(N^q) < O(\sqrt{N}).
    \label{eqn-alpha-threshold}
\end{equation}

For a discussion on using higher-order approximations, see Appendix \ref{sec:higher-order-approximations}. 

A couple of points to conclude this discussion:
\begin{itemize}
    \item As we can see from the functional form of the Hessian, we only need to worry when two $s_g$ are identical or close to each other given the sample size. Only here is its magnitude very large. Elsewhere in the domain, the Hessian is close to zero. 
    \item Thus, we only need to worry about the error of the delta method when some of the true maxes (or true mins) are identical. If this is the case, we need to control $\alpha$ to be strictly lower than $\sqrt{N}$. Elsewhere, $\alpha$ can be as large as desired (which means we could directly use the true min and max functions, because they are differentiable there). 
    \item Where max is not differentiable, the second order delta method term adds a Chi-squared term to the asymptotic distribution. This adds a positive bias to the distribution (and negative in the case of min) when $\alpha$ is too big. 
    \item Lowering $\alpha$ means changing the goalpost in terms of what we want to estimate. That is, instead of estimating the true max, we choose to estimate its smooth approximation instead.
    \item For a given $\alpha$, we may use a second-order Taylor expansion to incorporate the Chi-squared term and better model the asymptotic bias.
\end{itemize}

\section{Inference on DP}
\label{sec:inference-on-DP}

Let $\Sigma^\prime$ denote the variance-covariance matrix in equation \ref{clt-minmax}. Using the function $g(a, b) = \frac{a}{b}$ and the delta method on equation \ref{clt-minmax}, we then get
\begin{equation}
    \sqrt{N} (\hat{\text{DP}} - \text{DP}_0) \stackrel{d}{\longrightarrow} \mathcal{N} \left(0, \nabla^\top g \cdot \Sigma^\prime \cdot \nabla g  \right),
\end{equation}

for $\nabla g (a, b) = [1/b, -a/b^2]^\top$ evaluated at $a = \text{rsmin}_\alpha \left\{ s_j \right\}$ and $b = \text{rsmax}_\alpha \left\{ s_j \right\}$. This gives us the desired asymptotic distribution of DP. 

Call this final variance $\sigma^2_{\text{DP}}$. Then an appropriate 95\% two-sided confidence interval would be 
\begin{equation*}
    \text{DP}_0 \in \left[ \hat{\text{DP}} \pm 1.96 \frac{\hat{\sigma}_{\text{DP}}}{\sqrt{N}} \right].
\end{equation*}

\subsection{Importance of an appropriate rate of $\alpha$}
As discussed above, where max and min functions are differentiable, we have perfect asymptotic normality everywhere. However, there can be significant nonnormality when the true effects are identical (when the estimands are equal for at least some groups, i.e. where the min or max functions are nondifferentiable). We present a simple visualization to see how.\\

We work with two equally sized groups from a sample of size $N=1,000,000$. First, we start with distinct true success rates at 10\% and 5\%. We then plot the distribution of $\sqrt{N} (\hat{\text{DP}} - \text{DP}_0)$ for soft approximations of the DP function (centered on the soft version of the ground truth) and the ratio of the true min and max functions. We plot these histograms using 100,000 simulations each in Figure \ref{fig:distinct_means}. As expected, we find perfect asymptotic normality everywhere. \\

\begin{figure}[H]
    \centering
    \includegraphics[scale=0.6]{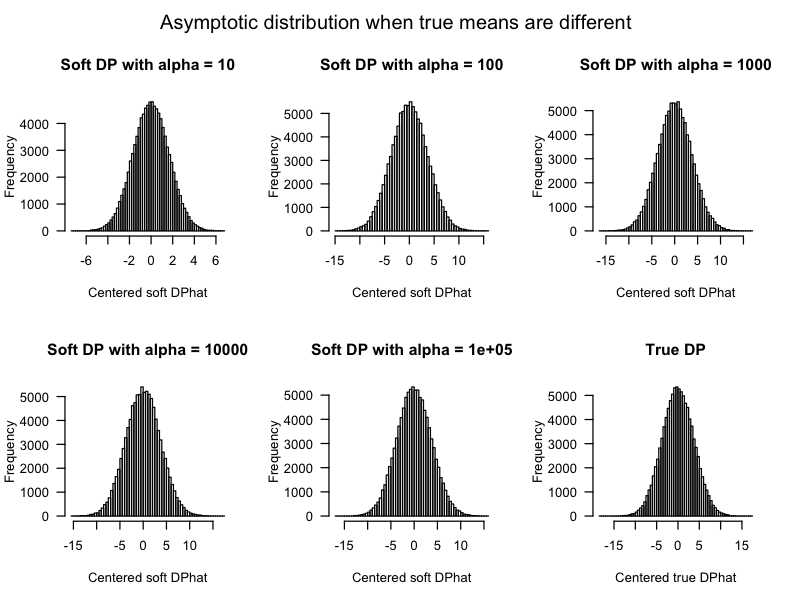}
    \caption{When true effects are distinct, we have asymptotic normality everywhere because the true min-over-max function itself is smooth and there is no need for the approximation. Here soft DP refers to the approximation estimator with degree of approximation $\alpha$ as indicated above, and true DP is the min-over-max estimator itself.}
    \label{fig:distinct_means}
\end{figure}

Next, we plot the case where the true means are in fact equal (at 10\%) in Figure \ref{fig:equal_means}. Remember that this is the case where we no longer have smoothness in the DP function to claim clean delta method inference. We find that we have asymptotic normality as long as $\alpha$ is small, specifically as $\alpha < O(\sqrt{N})$. This is indeed what our theory predicted (Look at equation \ref{eqn-alpha-threshold}). When $\alpha$ is large, we see significant left-skewness is the asymptotic distribution. This arises from the fact that, when the true means are equal, the min estimator is mechanically always lower than the max estimator and no longer identifies the ``true'' min group.

\begin{figure}[H]
    \centering
    \includegraphics[scale=0.5]{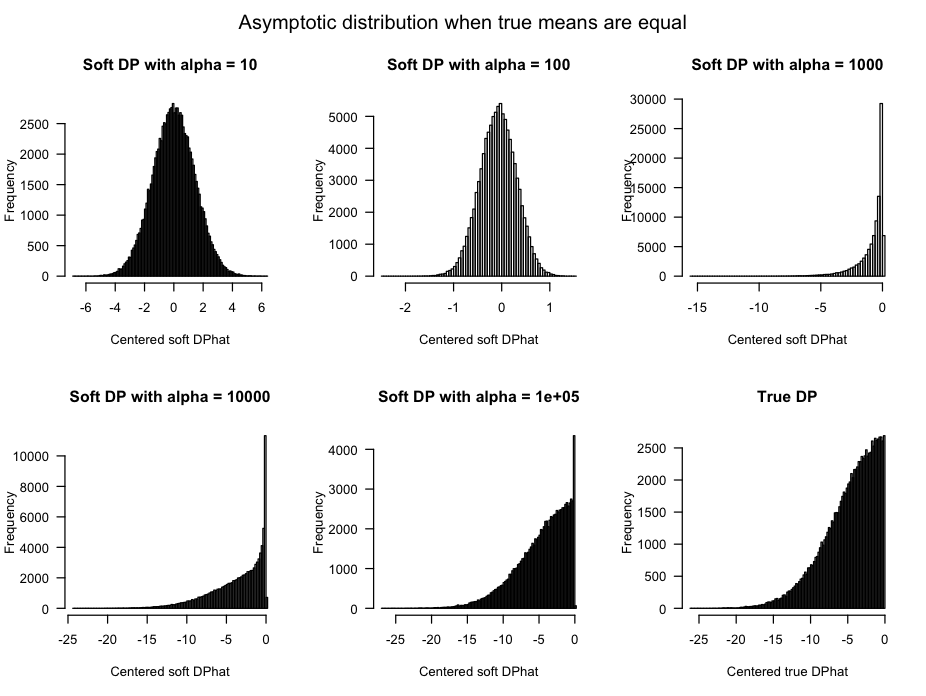}
    \caption{$N=1,000,000$ here. When true effects are identical, we have asymptotic normality as long as $\alpha < O(\sqrt{N})$. Once we have passed that $\alpha$ threshold—1,000 here—the distribution becomes sharply left-skewed for both the approximate soft DP estimator and the true DP estimator itself.}
    \label{fig:equal_means}
\end{figure}

\subsection{Testing if $\text{DP} < 0.8$}
Recall that $\text{DP} \in [0, 1]$. Therefore testing if $\text{DP} \in [0.8, 1.25]$ (the Equal Employment Opportunity Commission thresholds for parity) is the same as testing if $\text{DP} < 0.8$. We formulate the hypotheses as
\begin{align*}
    H_0: & \text{  DP} \geq 0.8 \\
    H_1: & \text{  DP} < 0.8.
\end{align*}

The Z-statistic of interest would be 
\begin{equation*}
    Z = \frac{\hat{\text{DP}} - 0.8}{\hat{\sigma}_{\text{DP}} / \sqrt{N}},
\end{equation*}

and we test if $Z < -1.645$, the 5\% quantile of the standard Normal distribution. In other words, the p-value for this one-sided test is 
\begin{equation*}
    \text{p-val} = \Phi(Z),
\end{equation*}
for a standard normal CDF, $\Phi(.)$, and we reject the null if this p-value is under 0.05. 

\subsection{Testing if $\text{DP}_A \neq \text{DP}_B$ in an A/B test}
Suppose we run an A/B test with two variants, $A$ and $B$. We seek to test if the experiment caused a change in the equity distribution across variants. Then we have two sets of statistics, $(\hat{DP}_A, \frac{\hat{\sigma}_A}{\sqrt{N_A}})$ for variant $A$ and $(\hat{DP}_B, \frac{\hat{\sigma}_B}{\sqrt{N_B}})$ for variant $B$. The two-sided hypotheses are:
\begin{align*}
    H_0: & \text{  DP}_B = \text{DP}_A \\
    H_1: & \text{  DP}_B \neq \text{DP}_A.
\end{align*}

The test statistic now is 
\begin{equation*}
    Z = \frac{\hat{\text{DP}}_B - \hat{\text{DP}}_A}{\sqrt{\frac{\hat{\sigma}_B^2}{N_B} + \frac{\hat{\sigma}_A^2}{N_A}}},
\end{equation*}
and we test if $|Z| < 1.96$. If not, we reject the null. The p-value for this test is $2 (1 - \Phi(|Z|))$.

If a one-sided test is desired, say to test of variant $B$ is more equal than variant $A$, like so
\begin{align*}
    H_0: & \text{  DP}_B \leq \text{DP}_A \\
    H_1: & \text{  DP}_B > \text{DP}_A,
\end{align*}
the test statistic remains the same, but now we test if $Z < 1.645$. If not, we reject the null. The p-value for this test is $1 - \Phi(Z)$. As a reminder, a one-sided test is more powerful because we work under the assumption that $B$ would only be significantly better than $A$, and do not test the opposite direction. 


\section{Other considerations}
\label{sec:other-considerations}
\subsection{Limited dependent variables}
In our analysis, we worked under a Bernoulli model where we assumed iid data and that we observe successes for each trial. In an online jobs marketplace, this means we observe whether each job application was successful or not. That is, we assume we can see all job offers a member has received. In the real world, this may not be the case. We may only observe all applications a member has submitted and the final job they accepted (and updated on their profile). \\

Consider this new model: 
\begin{itemize}
    \item A member $i$ belongs to group $g$, which we observe, and has job success probability of $s_g$.
    \item Additionally, we observe that they have applied to $N_i$ jobs.
    \item Naturally, the number of job offers they obtain $Y_i^* \sim \text{Binom}(s_g, N_i)$. However, we do not observe this variable.
    \item Instead, we observe whether the member has a new job or not at the end of the study period. That is, we observe $Y_i = 1\left\{Y_i^* \geq 1 \right\}$.
\end{itemize}
  Then, from the binomial distribution, we know that $Y_i$ is 1 with probability $1 - (1 - s_g)^{N_i}$ and 0 otherwise. That is,
  \begin{equation*}
      Y_i \sim \text{Bernoulli} \left(1 - (1 - s_g)^{N_i} \right).
  \end{equation*}
  Thus, within each group, the log likelihood function is 
  \begin{equation*}
      l(s_g; \left\{Y_i, N_i \right\}) = \sum_i \left( Y_i \log(1 - (1 - s_g)^{N_i}) + (1-Y_i)N_i \log(1-s_g) \right).
  \end{equation*}
  
  We can then obtain the MLE $\hat{s}_g$, which is asymptotically normal. Given this asymptotic normality, we can reapply the machinery we developed above to do inference on DP. 
  
\section{A small-sample critique of the DP estimator}
\label{sec:small-sample-critique}

For the estimand, $DP_0 = \frac{\min_g s_g}{\max_g s_g}$, the sample ratio estimator, $\frac{\min_g \hat{s}_g}{\max_g \hat{s}_g}$ is natural. Indeed, we showed above that it is consistent and has asymptotic normality. The problem of this estimator, though, is that it assumes there is an unambiguous min and max and the only uncertainty to be modeled is in the precise magnitude of this min and max. 

In finite sample, when confidence intervals of various $s_g$ may overlap, there is additional uncertainty in whether we picked the right $g$ as the argmin or argmax. This means $\min_g \hat{s}_g$ is biased downward and $\max_g \hat{s}_g$ is biased upward (To see this, recall that the max function is convex. Now apply Jensen's inequality to show the result). 

Practically, if the true $DP_0 = 1$, then the estimator $\hat{\text{DP}}$ is biased downward and will have less than perfect coverage. We may also have poor coverage when some groups are especially small and thus have their sample means estimated with much imprecision.

As an alternative to the sample min and max estimators, we may use cross-fitted estimators \citep{van2013estimating}. They work as follows:
\begin{enumerate}
    \item Split the data into $K$ folds. 
    \item For each $k \in \left\{1, 2, ..., K \right\}$, find the argmax and argmin groups on the $-k$ folds, i.e. on all the folds other than the $k$-th one. \label{sample-split1}
    \item Using the argmin group in the numerator and the argmax group in the denominator, obtain an estimate of DP on the $k$-th fold. \label{sample-split2}
    \item Repeat steps \ref{sample-split1} -- \ref{sample-split2} to obtain DP estimates on all $K$ folds.
    \item The average of these $K$ estimates is our final estimate.
\end{enumerate}

  
\section{Conclusion}
\label{sec:conclusion}
  
  We conclude by reiterating that the DP inference machinery we developed here is much more generic than the motivating Binomial model suggests. It can be applied for the min and maxes of any parameters, for whose estimators we have asymptotic normality, which is several of the most popular econometric estimators.


\bibliographystyle{apalike} 
\bibliography{references}

\pagebreak

\appendix 
\section{Appendix}
\subsection{Behavior of $\hat{\text{DP}}$ estimator where max function is not differentiable}
\label{appendix-DP-nondifferentiable}
\begin{prop}
The true max estimator is neither asymptotically unbiased not normally distributed wherever the max function is non-differentiable.
\end{prop}

Consider the simplest case where the max function is not differentiable. We have exactly 2 groups ($G$ = 2), equally sampled (i.e. $p_1 = p_2 = p = 0.5$) and they both have equal success probabilities $s_1 = s_2 = s$. 

A simple way to see the result: in large sample, 
\begin{equation*}
    \hat{s}_g \sim s + \sqrt{\frac{s(1-s)}{N/2}} Z_g,
\end{equation*}
for iid standard normal variables, $Z_g$. Thus,
\begin{equation*}
    \begin{split}
        \hat{s}_{(2)}  & := \max (\hat{s}_1, \hat{s}_2) \\
        & \sim s + \sqrt{\frac{s(1-s)}{N/2}} \max(Z_1, Z_2).
    \end{split}
\end{equation*}
$\max(Z_1, Z_2)$ is not normally distributed, and at any rate, is not centered at 0. This demonstrates our result. For a more formal treatment, see the proof below.

\begin{proof}
The sample means, $\hat{s}_g$ are iid for both groups at any given sample size, $N$, and they follow the distribution $\frac{1}{M} \text{Binom}(M, s)$, where $M = N/2$. Call this distribution $F_M(x)$.

Then the distribution of $\hat{s}_{(2)} = \max_g \hat{s}_g = F_M(x)^G = F_M(x)^2$, given that the two sample means are iid. Remember that this sample max estimator is estimating the true max, which is simply $s$.

From the CLT of the Binomial distribution, we know that 
\begin{equation*}
    \sqrt{\frac{M}{s(1-s)}} \left( \hat{s}_g - s \right) \stackrel{d}{\longrightarrow} \mathcal{N} (0, 1) \qquad \forall g \in \left\{1, 2 \right\}.
\end{equation*}
Put another way, we have 
\[
    P \left( \sqrt{\frac{M}{s(1-s)}} \left( \hat{s}_g - s \right) \leq x \right) \rightarrow \Phi(x) \qquad \forall x \in \mathbb{R},
\]
where $\Phi$ is the CDF of the standard normal distribution. This implies
\begin{equation*}
    P \left( \hat{s}_g \leq \sqrt{\frac{s(1-s)}{M}} x + s \right) \rightarrow \Phi(x).
\end{equation*}
That is,
\begin{equation}
   \lim_{M \rightarrow \infty} F_M \left( \sqrt{\frac{s(1-s)}{M}} x + s \right) = \Phi(x).
\end{equation}

We now want the distribution of $\sqrt{N} \left( \hat{s}_{(2)} - s \right)$, since we know $\hat{s}_{(2)} \stackrel{p}{\rightarrow} s$. So, we want 
\begin{equation*}
\begin{split}
        P\left(\sqrt{N}( \hat{s}_{(2)} - s) \leq x \right) & = P\left(\sqrt{2M}( \hat{s}_{(2)} - s) \leq x \right) \\
    & = P\left( \hat{s}_{(2)} \leq \frac{x}{\sqrt{2M}} + s \right) \\
    & = F_M \left( \frac{x}{\sqrt{2M}} + s \right)^2 \\
    & = F_M \left( \sqrt{\frac{s(1-s)}{M}} \left( \frac{x}{\sqrt{2s(1-s)}} \right) + s \right)^2 \\
    & \longrightarrow \Phi \left( \frac{x}{\sqrt{2s(1-s)}} \right)^2.
\end{split}
\end{equation*}
That is, $\sqrt{N} \left( \hat{s}_{(2)} - s \right)$ is \textit{not} asymptotically normal and has a strictly positive mean. 
\end{proof}

\begin{prop}
The true DP estimator is asymptotically one-tailed. That is, $\sqrt{N}(\hat{\text{DP}} - \text{DP}_0) \leq 0$ almost surely. Thus, the true DP estimator is neither asymptotically unbiased nor normally distributed. 
\end{prop}

\begin{proof}
As above, consider the case where we have two equally-sized groups with equal means:
\begin{equation*}
    \hat{s}_g \sim s + \frac{\sigma}{\sqrt{N/2}} Z_g.
\end{equation*}

Now, 
\begin{equation*}
    \begin{split}
        \sqrt{N} (\hat{\text{DP}} - \text{DP}_0) & = \sqrt{N} \left( \frac{\min_g \hat{s}_g}{\max_g \hat{s}_g} - 1 \right) \\
        & = \sqrt{N} \left( \frac{s + \frac{\sigma}{\sqrt{N/2}} \min_g Z_g}{s + \frac{\sigma}{\sqrt{N/2}} \max_g Z_g} - 1\right) \\
        & =  \sqrt{2}\sigma \left( \frac{\min_g Z_g - \max_g Z_g}{s + \frac{\sigma}{\sqrt{N/2}} \max_g Z_g} \right) \\ 
        & \stackrel{d}{\longrightarrow} \frac{\sqrt{2}\sigma}{s} \left( \min_g Z_g - \max_g Z_g \right).
    \end{split}
\end{equation*}

The random variable $\min_g Z_g - \max_g Z_g \leq 0$ almost surely. Further, in the two group case, we can simplify this further to 
\begin{equation*}
    -\frac{\sqrt{2} \sigma}{s} \left| Z_1 - Z_2 \right|,
\end{equation*}
which follows a Half-normal distribution on the negative reals.
\end{proof}

\subsection{Higher order approximations}
\label{sec:higher-order-approximations}
It is tempting to use a higher order terms of the Taylor expansion to get an even better approximation of the function. Suppose we used a second order Taylor expansion of the form
\begin{equation*}
    g(\hat{s}) = g(s) + \frac{\nabla g(s)}{1!} (\hat{s} - s) + \frac{\nabla^2 g(s)}{2}(\hat{s} - s)^2 + \underbrace{R_2(\hat{s})}_{\text{Remainder}}.
\end{equation*}
Again, consider the case of the $g$ being the rsmax function. Its third derivative has the form
\begin{equation*}
\begin{split}
    \frac{\partial^3 g}{\partial a_i \partial a_j \partial a_k} & = \alpha \cdot \underbrace{(\delta_{ij} - \mathcal{S}_j)\frac{\partial \mathcal{S}_i}{\partial a_k}}_{O_p(N^q)} - \alpha \cdot \underbrace{\mathcal{S}_i \frac{\partial \mathcal{S}_j}{\partial a_k}}_{O_p(N^q)} \\
    & = O_p(\alpha N^q) \\
    & = O_p(N^{2q}).
\end{split}
\end{equation*}

So the remainder approximation error term, $\sqrt{N} R_2$, which is now 
\begin{equation*}
    \sqrt{N} \cdot \underbrace{\frac{\nabla^3 g(s^{*})}{3!}}_{O_p(N^{2q})} \cdot \underbrace{(\hat{s}-s)^3}_{O_p(\frac{1}{N^{3/2}})}.
\end{equation*}
Thus, again we require $q < 1/2$ for this approximation error term to die down with $N$, meaning $\alpha < O(\sqrt{N})$ as before. In other words, the second order remainder does not buy us a larger $\alpha$ but for a given $\alpha$, it does reduce the remainder error by a polynomial order of magnitude. Observe what happens in the general case. Higher order derivatives become more and more unstable in the form of 
\begin{equation*}
    \nabla^k g(s) = O(\alpha^{k-1}) = O(N^{q(k-1)}).
\end{equation*}

This means the $k$th order Taylor expansion will have remainder error of the order
\begin{equation}
    O_p(R_k(\hat{s})) = O_p \left(\frac{1}{N^{(1/2-q)k}} \right).
\end{equation}
That is, for a given $\alpha$ (or a given $q$), we can improve approximation error with higher order terms. \\

\end{document}